\documentclass[11pt,a4paper,english]{amsart}

\usepackage[T1]{fontenc}

\usepackage{amsmath,amsfonts,amssymb}
\usepackage{mathrsfs}
\usepackage{dsfont}

\usepackage{tikz}
\usetikzlibrary{circuits.ee.IEC}

\usepackage{amsthm}

\theoremstyle{plain}
\newtheorem{theorem}{Theorem}
\newtheorem{statement}{Statement}
\newtheorem{lemma}{Lemma}
\newtheorem{corollary}{Corollary}

\theoremstyle{definition}
\newtheorem{example}{Example}
\newtheorem{definition}{Definition}

\theoremstyle{remark}
\newtheorem{remark}{Remark}

\numberwithin{equation}{section}

\newcommand{\E}{{\mathrm{e}}}

\newcommand{\br}[2]{\{ #1 , #2 \}}

\newcommand{\Qt}{\widetilde{Q}}

\newcommand{\Hil}{\mathscr{H}}
\newcommand{\Hh}{\hat{\mathscr{H}}}
\newcommand{\Kh}{\hat{\mathscr{K}}}
\newcommand{\K}{\mathscr{K}}
\newcommand{\Ih}{\hat{\mathscr{I}}}
\newcommand{\Is}{\mathscr{I}}
\newcommand{\Yr}{\hat{\Psi}}
\newcommand{\MC}{\mathrm{MC}(\Hh)}

\newcommand{\coh}{\mathcal{H}}

\newcommand{\even}{_{\mathrm{even}}}
\newcommand{\odd}{_{\mathrm{odd}}}

\newcommand{\ord}[1]{^{(#1)}}

\newcommand{\dc}{\partial_{c}}
\newcommand{\dq}{\partial_{\theta}}

\DeclareMathOperator{\End}{End}
\DeclareMathOperator{\im}{im}
\DeclareMathOperator{\ad}{ad}

\begin{document}

\title{Maurer-Cartan methods in perturbative quantum mechanics}

\author{Andrey Losev}

\address{Shanghai Institute for Mathematics and Interdisciplinary Sciences,
Building 3, 62 Weicheng Road, Yangpu District, Shanghai, 200433, China}
\email{aslosev2@yandex.ru}

\author{Tim Sulimov}

\address{St. Petersburg Department of Steklov Mathematical Institute of
Russian Academy of Sciences, 27 Fontanka, St. Petersburg, 191023, Russia}
\email{optimus260@gmail.com}

\maketitle

\begin{abstract}
	We reformulate the time-independent Schr\"odinger equation as a
	Maurer-Cartan equation on the superspace of eigensystems of the former
	equation. We then twist the differential so that its cohomology becomes the
	space of solutions with a set energy. A perturbation of the Hamiltonian
	corresponds to a deformation of the twisted differential, leading to a
	simple recursive relation for the eigenvalue and eigenfunction corrections.
\end{abstract}

\tableofcontents
\newpage

\section{Introduction}%
\label{sec:introduction}

In~\cite{losev2023} we reformulated the perturbative one-dimensional quantum
scattering problem in the language of homological algebra. In particular, we
verified that the space of solutions of the free problem
\begin{equation}\label{eq:free_qm}
	H_0 \psi = E \psi
\end{equation}
with fixed $E > 0$ can be viewed as the cohomology group of the differential
\begin{equation}\label{eq:free_Q}
	Q_0 = c (H_0 - E),
\end{equation}
where $c$ is a Grassmann variable. We then proceeded to apply homotopy transfer
to a deformed differential
\begin{equation}
	Q = Q_0 + Q_1 = c (H_0 - E) + c V
\end{equation}
using the homotopy $h$ for $Q_0$. The induced differential turned out to contain
the on-shell scattering K-matrix.

This approach relies on the degeneracy of energy levels and the isospectrality
of scattering problems, i.e., for every positive energy $E$ there exists a
solution of the full problem
\begin{equation}
	H_0 \psi + V \psi = E \psi
\end{equation}
that tends to two separate solutions of~\eqref{eq:free_qm} as $x \to \pm\infty$.
Thus, the treatment of bound states (especially non-degenerate ones) under small
perturbations requires a different method, as the energy levels shift and the
corrections to the original wave functions place them outside of the $Q_0$
cohomology.

In this paper we continue the homological algebra approach to quantum mechanics
by tackling the time-independent Schr\"odinger equation
\begin{equation}\label{eq:0schr}
	H \psi = E \psi,
\end{equation}
where the energy $E$ is no longer fixed but rather an unknown on par with the
wave function $\psi$.

The main idea is presented in Section~\ref{sec:schr_eq_in_mc_form} where we
combine the eigenvector with its eigenvalue into a single object $\Psi$ and
view~\eqref{eq:0schr} as a quadratic equation in $\Psi$. Following the
Maurer-Cartan approach to deformation problems~\cite{shadrin2022} we show that
this equation can be written as a Maurer-Cartan equation
\begin{equation}\label{eq:0mc}
	Q \Psi + \tfrac{1}{2} \br{\Psi}{\Psi} = 0
\end{equation}
in an appropriate differential graded Lie algebra, where the differential $Q$
depends on the Hamiltonian $H$ but the bracket does not. We also write down the
corresponding gauge group and apply the twisting procedure using a given
Maurer-Cartan element $\Psi_0$, i.e., a solution of~\eqref{eq:0schr}.

In Section~\ref{sec:perturbation_theory} we take the next step by deforming the
differential. This allows us to study the corrections to the eigensystem
under a small perturbation of $H$. In particular, we effortlessly acquire a
simple recurrence relation for these corrections. In addition, we show how these
corrections can be presented as a sum over tree diagrams.

Finally, in Section~\ref{sec:conclusion} we discuss how the Maurer-Cartan
approach can be applied to Hamiltonians with known symmetries and degenerate
spectra.

\section{Schr\"odinger equation in Maurer-Cartan form}%
\label{sec:schr_eq_in_mc_form}

\subsection{The superspace of eigensystems}%
\label{sub:the_superspace_of_eigensystems}

Consider a quantum mechanical problem with Hamiltonian $H$ acting on the Hilbert
space $\Hil$. The time-independent Schr\"odinger equation
\begin{equation}
	H \psi = E \psi, \quad \psi \in \Hil, E \in \mathbb{R},
\end{equation}
is just an eigenproblem for a self-adjoint operator $H$.

For the purpose of unifying the eigenfunctions of $H$ with the corresponding
eigenvalues we introduce the \emph{superspace of eigensystems}
\begin{equation}\label{eq:eigsys}
	\Hh = \left(\mathbb{C}[\theta, c] \otimes \Hil \right) \oplus \mathbb{C}[c],
\end{equation}
where $\theta$ and $c$ are Grassmann variables:
\begin{equation}
	\theta^2 = c^2 = 0, \quad [\theta, c]_{+} = \theta c + c \theta = 0.
\end{equation}
Despite looking contrived, this superspace is a natural extension of the
superspace used in~\cite{losev2023}. In particular, $c$ is the <<ghost>>
corresponding to $H$ viewed as a symmetry of itself. However, the meaning of
$\theta$ is not clear yet.

In what follows <<hats>> over spaces will signify that they are in fact
superspaces. The general form of an element $\Psi \in \Hh$ is
\begin{equation}
	\Psi =
	\begin{pmatrix}
	\psi_1 + \theta \psi_2 + c \psi_3 + c \theta \psi_4 \\
	E_1 + c E_2
	\end{pmatrix}
\end{equation}
and $\Hh$ thus has a natural $\mathbb{Z}$ grading
\begin{equation}
	\Hh = \Hh_0 \oplus \Hh_1 \oplus \Hh_2
\end{equation}
based on total degree in $\theta$ and $c$. To keep track of signs it is also
convenient to use the $\mathbb{Z}_2$ grading
\begin{equation}
	\Hh = \Hh\even \oplus \Hh\odd.
\end{equation}

Next we turn $\Hh$ into a graded Lie algebra by defining the following bracket:
\begin{equation}
	\br{\Psi}{\Phi} =
	\begin{pmatrix}
		\Psi^{\mathrm{T}} \epsilon \Phi \\
		0
	\end{pmatrix}
	, \quad \epsilon =
	\begin{pmatrix}
		0 & 1 \\
		-1 & 0
	\end{pmatrix}.
\end{equation}
\begin{example}
	A bracket of two generic elements in $\Hh_0$ is
	\begin{equation}
		\br{\Psi}{\Phi} =
		\br{
		\begin{pmatrix}
			\psi \\
			E
		\end{pmatrix}
		}{
		\begin{pmatrix}
			\varphi \\
			\mathcal{E}
		\end{pmatrix}
		} =
		\begin{pmatrix}
			\mathcal{E} \psi - E \varphi \\
			0
		\end{pmatrix}.
	\end{equation}
\end{example}
\begin{example}
	A bracket of two generic elements in $\Hh_1$ is
	\begin{equation}
		\br{\Psi}{\Phi} =
		\br{
		\begin{pmatrix}
			\theta \psi_1 + c \psi_2 \\
			c E
		\end{pmatrix}
		}{
		\begin{pmatrix}
			\theta \varphi_1 + c \varphi_2 \\
			c \mathcal{E}
		\end{pmatrix}
		} = -
		\begin{pmatrix}
			c \theta (\mathcal{E} \psi_1 + E \varphi_1) \\
			0
		\end{pmatrix}.
	\end{equation}
\end{example}
\begin{remark}
	This bracket is \emph{graded-antisymmetric}, i.e., it is symmetric if both
	elements are odd and antisymmetric if at least one is even.
\end{remark}
We also define
\begin{equation}
	Q =
	\begin{pmatrix}
		c H & 0 \\
		0 & 0
	\end{pmatrix}
	\in \End{\Hh}.
\end{equation}
It can be easily checked that $Q$ is a differential of degree $1$, i.e., it
squares to zero and is a derivation of the bracket $\br{\cdot}{\cdot}$. Thus,
$\Hh$ is a differential graded Lie algebra.

\subsection{Maurer-Cartan equation}%
\label{sub:maurer_cartan_equation}

We now restrict our attention to the subspace $\Hh_1$ and reformulate the
eigenvalue problem in Maurer-Cartan terms.

\begin{definition}
	A Maurer-Cartan element is a solution of the Maurer-Cartan equation
	\begin{equation}\label{eq:mc}
		Q \Psi + \tfrac{1}{2} \br{\Psi}{\Psi} = 0
	\end{equation}
	in $\Hh_1$. The set of Maurer-Cartan elements is denoted by $\MC$.
\end{definition}

\begin{statement}
	$\Psi \in \MC$ if and only if
	\begin{equation}\label{eq:equiv}
		\Psi =
		\begin{pmatrix}
			\theta \psi + c \phi \\
			c E
		\end{pmatrix},
	\end{equation}
	where $\phi$ is any function and $H \psi = E \psi$.
\end{statement}

The freedom in the choice of $\phi$ in~\eqref{eq:equiv} seems irrelevant to
the eigenproblem.  Additionally, different normalizations of $\psi$ correspond
to different elements of $\MC$, which so far has no internal structure to
account for this multiplicity. Both of these issues are reconciled by the gauge
group associated to $\Hh$. The following definition and theorem are adapted
directly from~\cite{shadrin2022}.

\begin{definition}
	The gauge group $\Gamma$ associated to $\Hh$ is the group obtained from
	$\Hh_0$ via the Baker-Campbell-Hausdorff formula. Namely, $1_\Gamma = 0$ and
	for $\Psi, \Phi \in \Hh_0$ their product is
	\begin{equation}
		\Psi \cdot \Phi = \ln(\E^\Psi \E^\Phi) = \Psi + \Phi
		+ \tfrac{1}{2} \br{\Psi}{\Phi}
		+ \tfrac{1}{12} \br{\Psi}{\br{\Psi}{\Phi}}
		+ \ldots.
	\end{equation}
\end{definition}

\begin{remark}
	It can be easily seen that
	\begin{equation}
		\begin{pmatrix}
			\psi \\
			E
		\end{pmatrix}
		\cdot
		\begin{pmatrix}
			\varphi \\
			\mathcal{E}
		\end{pmatrix}
		=
		\begin{pmatrix}
			f_1(E,\mathcal{E}) \psi + f_2(E,\mathcal{E}) \varphi \\
			E + \mathcal{E}
		\end{pmatrix}.
	\end{equation}
	Although $f_1$ and $f_2$ can be computed, they are of no direct interest
	to us.
\end{remark}

\begin{theorem}
	The formula
	\begin{equation}
		\Phi . \Psi = \frac{\mathds{1} - \E^{\ad_\Phi}}{\ad_\Phi} (Q \Phi)
		+ \E^{\ad_\Phi} \Psi,
	\end{equation}
	where $\ad_\Phi = \br{\Phi}{-}$, defines a left action of the gauge group
	$\Gamma$ on $\MC$.	
\end{theorem}

Given $\Phi = (\varphi ~ \mathcal{E})^{\mathrm{T}}$ one readily gets
\begin{equation}
	\ad_\Phi =
	\begin{pmatrix}
		- \mathcal{E} & \varphi \\
		0 & 0
	\end{pmatrix}
	, \qquad
	\E^{\ad_\Phi} =
	\begin{pmatrix}
		\E^{-\mathcal{E}} & \frac{1 - \E^{-\mathcal{E}}}{\mathcal{E}} \varphi \\
		0 & 1
	\end{pmatrix},
\end{equation}
as well as
\begin{equation}
	\frac{\mathds{1} - \E^{\ad_\Phi}}{\ad_\Phi} (Q \Phi) =
	\begin{pmatrix}
		c \frac{\E^{-\mathcal{E}} - 1}{\mathcal{E}} H \varphi \\
		0
	\end{pmatrix},
\end{equation}
producing the formula for the gauge action
\begin{equation}
	\Phi .
	\begin{pmatrix}
		\theta \psi + c \phi \\
		c E
	\end{pmatrix}
	=
	\begin{pmatrix}
		\theta \E^{-\mathcal{E}} \psi + c \left( \E^{-\mathcal{E}} \phi
		+ \frac{\E^{-\mathcal{E}} - 1}{\mathcal{E}} (H - E) \varphi \right) \\
		c E
	\end{pmatrix}.
\end{equation}

As we can see, using an appropriate gauge transformation we can fix the
normalization of $\psi$ as well as remove all $\im(H - E)$ components from
$\phi$. This still leaves $\phi \in \ker(H - E)$ not fixed. The meaning of this
freedom is unclear.

Once we have a solution of~\eqref{eq:mc} we can twist the differential:

\begin{corollary}
	Let $\Psi \in \Hh_1$ be a solution of~\eqref{eq:mc}, then
	\begin{equation}
		\Qt = Q + \br{\Psi}{\cdot}
	\end{equation}
	is also a differential. Moreover, if
	\begin{equation}
		\Psi =
		\begin{pmatrix}
			\theta \psi \\
			c E
		\end{pmatrix}
	\end{equation}
	then
	\begin{equation}
		\Qt =
		\begin{pmatrix}
			c (H - E) & \theta \psi \\
			0 & 0
		\end{pmatrix}.
	\end{equation}
\end{corollary}

\begin{corollary}
	Let $\Psi \in \Hh_1$ be a solution of~\eqref{eq:mc} and let $\Phi \in \Hh_1$
	be such that $\Psi + \Phi$ is another solution of~\eqref{eq:mc}. Then
	\begin{equation}
		\Qt \Phi + \tfrac{1}{2} \br{\Phi}{\Phi} = 0.
	\end{equation}
\end{corollary}

\subsection{Kernel and cohomology}%
\label{sub:kernel_and_cohomology}

Let us introduce some notation:
\begin{equation}
	H_E = H - E, \qquad
	\K = \ker{H_E}, \qquad
	\Kh = \ker{\Qt},
\end{equation}
\begin{equation}
	\Is = \im{H_E}, \qquad
	\Ih = \im{\Qt},
\end{equation}
where $\Qt$ is constructed using a nonzero $\psi \in \K$:
\begin{equation}
	\Qt =
	\begin{pmatrix}
		c (H - E) & \theta \psi \\
		0 & 0
	\end{pmatrix}.
\end{equation}

\begin{statement}
	$\Psi \in \Kh$ if and only if there exist $\psi_1, \psi_2 \in \K$ and
	$\varphi_1, \varphi_2 \in \Hil$ such that
	\begin{equation}\label{eq:ker_Qt}
		\Psi =
		\begin{pmatrix}
			\psi_1 + \theta \psi_2 + c \varphi_1 + c \theta \varphi_2 \\
			0
		\end{pmatrix}.
	\end{equation}
\end{statement}

\begin{proof}
	The <<if>> part is trivial. Now consider $\Qt$ acting on a generic element
	in $\Hh$:
	\begin{align}
		&
		\begin{pmatrix}
			c H_E & \theta \psi \\
			0 & 0
		\end{pmatrix}
		\begin{pmatrix}
			\psi_1 + \theta \psi_2 + c \varphi_1 + c \theta \varphi_2 \\
			E_1 + c E_2
		\end{pmatrix}
		\\ \label{eq:im_Qt}
		&=
		\begin{pmatrix}
			\theta E_1 \psi + c H_E \psi_1 + c \theta (H_E \psi_2 - E_2 \psi) \\
			0
		\end{pmatrix}.
	\end{align}
	For this to be zero we necessarily have $E_1 = 0$, $H_E \psi_1 = 0$, and
	$H_E \psi_2 = E_2 \psi$. But the image of a self-adjoint operator is
	orthogonal to its kernel. Thus, $H_E \psi_2 = 0$ and consequently $E_2 = 0$.
	This completes the proof.
\end{proof}
Next we describe the cohomology of $\Qt$.
For simplicity and clarity we shall assume that $\K = \langle \psi \rangle$.
This means that the eigenstate $\psi$ is non-degenerate~--- a common occurrence
for physical systems.
\begin{corollary}
	The cohomology of $\Qt$ is
	\begin{equation}
		\coh_{\Qt} \cong \mathbb{C}[c] \otimes \K.
	\end{equation}
\end{corollary}
\begin{proof}
	Elements of $\Kh$ and $\Ih$ are described by~\eqref{eq:ker_Qt}
	and~\eqref{eq:im_Qt}, respectively. Comparing the two forms we can see that
	appropriate choice of $\psi_1$, $\psi_2$, $E_1$ and $E_2$
	in~\eqref{eq:im_Qt} allows us to acquire any element of
	\begin{equation}
		\begin{pmatrix}
			\theta \K + c \Is + c \theta \Hil \\
			0
		\end{pmatrix}.
	\end{equation}
	When we factor the closed elements~\eqref{eq:ker_Qt} by the exact elements
	of this space we are left with
	\begin{equation}
		\begin{pmatrix}
			\K + c \K \\
			0
		\end{pmatrix}.
	\end{equation}
\end{proof}

\subsection{Homotopy}%
\label{sub:homotopy}

We introduce a homotopy
\begin{equation}
	h =
	\begin{pmatrix}
		\dc G & 0 \\
		\dq Y & 0
	\end{pmatrix},
\end{equation}
where $G$ is defined as
\begin{equation}\label{eq:G}
	G\big|_{\K} = 0, \qquad
	G\big|_{\Hil \setminus \K} = H_E^{-1},
\end{equation}
and $Y$ is a projection functional onto $\psi$ in $\Hil$:
\begin{equation}
	Y(\varphi) = (\psi, \varphi).
\end{equation}
Here we assume that $\psi$ is appropriately normalized using the gauge group.

\begin{remark}
	This homotopy squares to zero, since $Y G \varphi = 0$ for any $\varphi$.
\end{remark}

\begin{lemma}
	The anti-commutator of $\Qt$ and $h$ is
	\begin{equation}
		[\Qt, h]_{+} = \mathds{1} - \Pi_{\coh_{\Qt}}.
	\end{equation}
\end{lemma}
\begin{proof}
	We have to keep track of what the individual components are acting on to do
	the computation:
	\begin{equation}
		\Qt h + h \Qt =
		\begin{pmatrix}
			c \dc H_E G + \theta \dq \psi Y & 0 \\
			0 & 0
		\end{pmatrix}
		+
		\begin{pmatrix}
			\dc c G H_E & \dc \theta G \psi \\
			\dq c Y H_E & \dq \theta Y \psi
		\end{pmatrix}.
	\end{equation}
	In particular,
	\begin{equation}
		(G \psi) (\mathcal{E}) = \mathcal{E} G (\psi) = 0
	\end{equation}
	by definition of $G$. For the same reason
	\begin{equation}
		H_E G = G H_E = 1 - \Pi_{\K},
	\end{equation}
	where $\Pi$ denotes the projection operator. Additionally
	\begin{equation}
		(Y H_E) (\varphi) = (\psi, H_E \varphi) =
		(H_E \psi, \varphi) = 0,
	\end{equation}
	and also
	\begin{equation}
		(\psi Y) (\varphi) = \psi (\psi, \varphi) =
		\Pi_{\psi} \varphi, \qquad
		(Y \psi) (\mathcal{E}) = \mathcal{E} Y(\psi) = \mathcal{E}.
	\end{equation}
	Combining all of this gives
	\begin{equation}
		[\Qt, h]_{+} =
		\begin{pmatrix}
			1 - \Pi_{\K} + \theta \dq \Pi_{\psi} & 0 \\
			0 & 1
		\end{pmatrix},
	\end{equation}
	which can be further simplified if we recall that $\Pi_{\psi} = \Pi_{\K}$,
	producing
	\begin{equation}
		[\Qt, h]_{+} =
		\begin{pmatrix}
			1 - \dq \theta \Pi_{\K} & 0 \\
			0 & 1
		\end{pmatrix}
		= \mathds{1} - \Pi_{\coh_{\Qt}}.
	\end{equation}
\end{proof}

\section{Perturbation theory}%
\label{sec:perturbation_theory}

Let us now develop perturbation theory using the formalism above.
We would like to solve the quantum mechanical problem
\begin{equation}\label{eq:pert_pr}
	H \psi = (H_0 + \lambda V) \psi = E \psi
\end{equation}
perturbatively in $\lambda$, assuming we can solve the non-perturbed problem
\begin{equation}\label{eq:nonpert_pr}
	H_0 \psi\ord{0} = E\ord{0} \psi\ord{0}.
\end{equation}
\begin{remark}
	The parameter $\lambda$ can be thought of as small, but is actually
	there only to keep track of the powers of $V$.
\end{remark}

\subsection{Non-perturbed problem}%
\label{sub:non_perturbed_problem}

In this section we only recall and adjust the notation from before. Let
\begin{equation}
	\Psi\ord{0} =
	\begin{pmatrix}
		\theta \psi\ord{0} \\
		c E\ord{0}
	\end{pmatrix}
\end{equation}
be a solution of
\begin{equation}\label{eq:mc_nonpert}
	Q_0 \Psi\ord{0} + \tfrac{1}{2} \br{\Psi\ord{0}}{\Psi\ord{0}} = 0,
\end{equation}
where
\begin{equation}
	Q_0 =
	\begin{pmatrix}
		c H_0 & 0 \\
		0 & 0
	\end{pmatrix}.
\end{equation}
We embed $\Psi\ord{0}$ into $Q_0$ by defining
\begin{equation}\label{eq:Qt0}
	\Qt_0 = Q_0 + \br{\Psi\ord{0}}{\cdot} =
	\begin{pmatrix}
		c (H_0 - E\ord{0}) & \theta \psi\ord{0} \\
		0 & 0
	\end{pmatrix}.
\end{equation}
The corresponding homotopy is
\begin{equation}
	h_0 =
	\begin{pmatrix}
		\dc G_0 & 0 \\
		\dq Y_0 & 0
	\end{pmatrix},
\end{equation}
where $G_0$ is defined as
\begin{equation}\label{eq:G}
	G_0\big|_{\K_0} = 0, \quad
	G_0\big|_{\Hil \setminus \K_0} = (H_0 - E\ord{0})^{-1}, \quad
	\K_0 = \ker(H_0 - E\ord{0})
\end{equation}
and $Y_0$ is a linear functional on $\Hil$:
\begin{equation}
	Y_0(\varphi) = (\psi\ord{0}, \varphi).
\end{equation}
We also have
\begin{equation}
	[\Qt_0, h_0]_{+} = \mathds{1} - \Pi_{\coh_{\Qt_0}},
\end{equation}
where
\begin{equation}\label{eq:coh_0}
	\coh_{\Qt_0} =
	\begin{pmatrix}
		\K_0 + c \K_0 \\
		0
	\end{pmatrix}.
\end{equation}

\subsection{Perturbation}%
\label{sub:perturbation}

It is clear that equation~\eqref{eq:pert_pr} can be rendered in Maurer-Cartan
form as
\begin{equation}\label{eq:mc_pert}
	Q \Psi + \tfrac{1}{2} \br{\Psi}{\Psi} = 0,
\end{equation}
where $\Psi \in \Hh_1$, $Q = Q_0 + \lambda Q_1$, and
\begin{equation}
	Q_1 =
	\begin{pmatrix}
		c V & 0 \\
		0 & 0
	\end{pmatrix}.
\end{equation}
In accordance with our original intention we search for $\Psi$ in the form
\begin{equation}\label{eq:split}
	\Psi = \Psi\ord{0} + \hat{\Psi}, \qquad
	\hat{\Psi} =
	\begin{pmatrix}
		\theta \hat{\psi} \\
		c \hat{E}
	\end{pmatrix},
\end{equation}
\begin{equation}\label{eq:Yr_series}
	\hat{\Psi} = \sum_{k = 1}^{\infty} \lambda^k \Psi\ord{k}
	= \sum_{k = 1}^{\infty} \lambda^k
	\begin{pmatrix}
		\theta \psi\ord{k} \\
		c E\ord{k}
	\end{pmatrix}.
\end{equation}
To fix the normalization we impose that
\begin{equation}\label{eq:norm}
	(\psi\ord{0}, \psi\ord{k}) = \delta_{k,0}.
\end{equation}

Substituting~\eqref{eq:split} into~\eqref{eq:mc_pert} we get
\begin{align}
	&{} Q_0 \Psi\ord{0} + \lambda Q_1 \Psi\ord{0} + Q_0 \Yr
	+ \lambda Q_1 \Yr \\
	&+ \tfrac{1}{2} \br{\Psi\ord{0}}{\Psi\ord{0}}
	+ \tfrac{1}{2} \br{\Psi\ord{0}}{\Yr}
	+ \tfrac{1}{2} \br{\Yr}{\Psi\ord{0}}
	+ \tfrac{1}{2} \br{\Yr}{\Yr} = 0,
\end{align}
which can be rearranged using~\eqref{eq:mc_nonpert},~\eqref{eq:Qt0}, and the
properties of the bracket:
\begin{equation}
	\Qt_0 \Yr = - \lambda Q_1 \Psi\ord{0} - \lambda Q_1 \Yr
	- \tfrac{1}{2} \br{\Yr}{\Yr}.
\end{equation}
Before dealing with $\Qt_0$ using the homotopy we must check that the
obstructions vanish, i.e., the r.h.s. is zero in $\coh_{\Qt_0}$. This is indeed
the case, as can be seen from the Grassmann structure: the r.h.s. has to be
proportional to $c \theta$, which puts it outside $\coh_{\Qt_0}$, given
by~\eqref{eq:coh_0}. Applying $h_0$ to the l.h.s. gives
\begin{align}
	h_0 \Qt_0 \Yr &= (1 - \Pi_{\coh_{\Qt_0}} - \Qt_0 h_0) \Yr \\
	&=
	\begin{pmatrix}
		1 - \dq \theta \Pi_{\K_0} & 0 \\
		0 & 1
	\end{pmatrix}
	\begin{pmatrix}
		\theta \hat{\psi} \\
		c \hat{E}
	\end{pmatrix}
	\\
	&\quad-
	\begin{pmatrix}
		c (H_0 - E\ord{0}) & \theta \psi\ord{0} \\
		0 & 0
	\end{pmatrix}
	\begin{pmatrix}
		\dc G_0 & 0 \\
		\dq Y_0 & 0
	\end{pmatrix}
	\begin{pmatrix}
		\theta \hat{\psi} \\
		c \hat{E}
	\end{pmatrix}
	\\
	&=
	\begin{pmatrix}
		\theta \hat{\psi} \\
		c \hat{E}
	\end{pmatrix}
	-
	\begin{pmatrix}
		\theta \Pi_{\K_0} \hat{\psi} \\
		0
	\end{pmatrix}
	\\
	&= (1 - \Pi_{\Kh_0}) \Yr = \Yr,
\end{align}
where the last equality holds due to~\eqref{eq:norm}. Consequently,
\begin{equation}\label{eq:recurrence}
	\Yr = - \lambda h_0 Q_1 \Psi\ord{0} - \lambda h_0 Q_1 \Yr
	- \tfrac{1}{2} h_0 \br{\Yr}{\Yr}.
\end{equation}

\subsection{Recurrence relation}%
\label{sub:recurrence_relation}

We can now substitute the series expansion~\eqref{eq:Yr_series} for $\Yr$ and
gather all the terms of order $k$ to acquire the following recurrence relation
\begin{equation}\label{eq:recurrence_k}
	\Psi\ord{k} = - h_0 Q_1 \Psi\ord{k-1}
	- \tfrac{1}{2} h_0 \sum_{m = 1}^{k-1} \br{\Psi\ord{m}}{\Psi\ord{k - m}}.
\end{equation}
In particular, for $k = 1, 2, 3$ they read
\begin{align}
	\Psi\ord{1} &= - h_0 Q_1 \Psi\ord{0}, \\
	\Psi\ord{2} &= - h_0 Q_1 \Psi\ord{1}
	- \tfrac{1}{2} h_0 \br{\Psi\ord{1}}{\Psi\ord{1}}, \\
	\Psi\ord{3} &= - h_0 Q_1 \Psi\ord{2}
	- h_0 \br{\Psi\ord{1}}{\Psi\ord{2}}.
\end{align}
These equations can be solved, producing expressions for eigenfunction and
eigenvalue corrections. For $k = 1$ we have
\begin{equation}
	\begin{pmatrix}
		\theta \psi\ord{1} \\
		c E\ord{1}
	\end{pmatrix}
	= -
	\begin{pmatrix}
		\dc G_0 & 0 \\
		\dq Y_0 & 0
	\end{pmatrix}
	\begin{pmatrix}
		c V & 0 \\
		0 & 0
	\end{pmatrix}
	\begin{pmatrix}
		\theta \psi\ord{0} \\
		c E\ord{0}
	\end{pmatrix}
	=
	\begin{pmatrix}
		- \theta G_0 V \psi\ord{0} \\
		c (\psi\ord{0}, V \psi\ord{0})
	\end{pmatrix},
\end{equation}
where we used the explicit form of $Y_0$.  Skipping some computations, for $k =
2$ we get
\begin{align}
	\begin{pmatrix}
		\theta \psi\ord{2} \\
		c E\ord{2}
	\end{pmatrix}
	&=
	\begin{pmatrix}
		\theta G_0 V G_0 V \psi\ord{0} \\
		- c Y_0 V G_0 V \psi\ord{0}
	\end{pmatrix}
	+
	\begin{pmatrix}
		- \theta (Y_0 V \psi\ord{0}) G_0 G_0 V \psi\ord{0} \\
		c (Y_0 V \psi\ord{0}) Y_0 G_0 V \psi\ord{0}
	\end{pmatrix}
	\\
	&=
	\begin{pmatrix}
		\theta \left( G_0 V G_0 V \psi\ord{0}
		- (\psi\ord{0}, V \psi\ord{0}) G_0^2 V \psi\ord{0} \right) \\
		- c (\psi\ord{0}, V G_0 V \psi\ord{0})
	\end{pmatrix},
\end{align}
where in the second equality we used the fact that $Y_0 G_0 \varphi = 0$ for any
$\varphi$.

\begin{remark}\label{rem:eig_cor}
	It is easy to see that only the first term in~\eqref{eq:recurrence_k}
	provides a correction to the eigenvalue due to $Y_0 G_0 = 0$.
\end{remark}

The simple structure of the energy corrections is broken in order $k
= 3$. In components:
\begin{align}
	\psi\ord{3} &= - G_0 V G_0 V G_0 V \psi\ord{0}
	+ (\psi\ord{0}, V \psi\ord{0}) G_0 V G_0^2 V \psi\ord{0} \\
	&+ (\psi\ord{0}, V G_0 V \psi\ord{0}) G_0^2 V \psi\ord{0}
	+ (\psi\ord{0}, V \psi\ord{0}) G_0^2 V G_0 V \psi\ord{0} \\
	&- (\psi\ord{0}, V \psi\ord{0})^2 G_0^3 V \psi\ord{0}, \\
	E\ord{3} &= (\psi\ord{0}, V G_0 V G_0 V \psi\ord{0})
	- (\psi\ord{0}, V \psi\ord{0}) (\psi\ord{0}, V G_0^2 V \psi\ord{0}).
\end{align}

We can see that this procedure reproduces the familiar formulas found in every
textbook on quantum mechanics (cf.~\cite{sakurai}).

\subsection{Corrections as sums over tree diagrams}%
\label{sub:corrections_as_sums_over_trees}

Let us introduce the following graphical notation for the various elements in
the recurrence relation:
\begin{equation}
	h_0 =
	\begin{tikzpicture}[semithick, circuit ee IEC,
		every resistor/.style={circuit symbol size=width 2 height .75}]
		\draw (0,0) to[resistor] (1,0);
	\end{tikzpicture}
	~, \quad Q_1 =
	\begin{tikzpicture}[semithick, circuit ee IEC, baseline=-3pt]
		\draw (0,0) -- ++(.4,0) node[contact] {} -- ++(.4,0);
	\end{tikzpicture}
	~, \quad \tfrac{1}{2}\br{-}{-} =
	\begin{tikzpicture}[semithick, baseline=-3pt]
		\draw (0,0) -- ++(.5,0) |- ++(.5,.25) ++(-.5,-.25) |- +(.5,-.25);
	\end{tikzpicture}
	~, \quad \Psi\ord{0} =
	\begin{tikzpicture}[semithick, baseline=-3pt]
		\draw[-o] (0,0) -- ++(.5,0);
	\end{tikzpicture}~.
\end{equation}
This allows us to rewrite the corrections in the form of tree diagrams:
\begin{align}
	\Psi\ord{1} &=
	\begin{tikzpicture}[semithick, baseline=-3pt, circuit ee IEC,
		every resistor/.style={circuit symbol size=width 2 height .75}]
		\draw[-o] (0,0) to[resistor] ++(1,0) -- ++(.1,0) node[contact] {} --
		++(.5,0);
	\end{tikzpicture}
	~, \\[3mm]
	\Psi\ord{2} &= 
	\begin{tikzpicture}[semithick, baseline=-3pt, circuit ee IEC,
		every resistor/.style={circuit symbol size=width 2 height .75}]
		\draw[-o] (0,0) to[resistor] ++(1,0) -- ++(.1,0) node[contact] {}
		to[resistor] ++(1,0) -- ++(.1,0) node[contact] {} -- ++(.5,0);
	\end{tikzpicture}
	\pm
	\begin{tikzpicture}[semithick, baseline=-3pt, circuit ee IEC,
		every resistor/.style={circuit symbol size=width 2 height .75}]
		\draw (0,0) to[resistor] ++(1,0) |- ++(.1,.25)
		coordinate (a) ++(-.1,-.25) |- +(.1,-.25) coordinate (b);
		\draw[-o] (a) to[resistor] ++(1,0) -- ++(.1,0) node[contact] {} --
		++(.5,0);
		\draw[-o] (b) to[resistor] ++(1,0) -- ++(.1,0) node[contact] {} --
		++(.5,0);
	\end{tikzpicture}
	~, \\[3mm]
	\Psi\ord{3}
	&= 
	\begin{tikzpicture}[semithick, baseline=-3pt, circuit ee IEC,
		every resistor/.style={circuit symbol size=width 2 height .75}]
		\draw[-o] (0,0) to[resistor] ++(1,0) -- ++(.1,0) node[contact] {
		}to[resistor] ++(1,0) -- ++(.1,0) node[contact] {}
		to[resistor] ++(1,0) -- ++(.1,0) node[contact] {} -- ++(.5,0);
	\end{tikzpicture}
	\pm
	\begin{tikzpicture}[semithick, baseline=-3pt, circuit ee IEC,
		every resistor/.style={circuit symbol size=width 2 height .75}]
		\draw (0,0) to[resistor] ++(1,0) -- ++(.1,0) node[contact] {}
		to[resistor] ++(1,0) |- ++(.1,.25)
		coordinate (a) ++(-.1,-.25) |- +(.1,-.25) coordinate (b);
		\draw[-o] (a) to[resistor] ++(1,0) -- ++(.1,0) node[contact] {} --
		++(.5,0);
		\draw[-o] (b) to[resistor] ++(1,0) -- ++(.1,0) node[contact] {} --
		++(.5,0);
	\end{tikzpicture}
	\\
	&\pm
	\begin{tikzpicture}[semithick, baseline=-3pt, circuit ee IEC,
		every resistor/.style={circuit symbol size=width 2 height .75}]
		\draw (0,0) to[resistor] ++(1,0) |- ++(.1,.25)
		coordinate (a) ++(-.1,-.25) |- +(.1,-.25) coordinate (b);
		\draw[-o] (a) to[resistor] ++(1,0) -- ++(.1,0) node[contact] {} --
		++(.5,0);
		\draw[-o] (b)
		to[resistor] ++(1,0) -- ++(.1,0) node[contact] {}
		to[resistor] ++(1,0) -- ++(.1,0) node[contact] {}
		-- ++(.5,0);
	\end{tikzpicture}
	\pm
	\begin{tikzpicture}[semithick, baseline=-3pt, circuit ee IEC,
		every resistor/.style={circuit symbol size=width 2 height .75}]
		\draw (0,0) to[resistor] ++(1,0) |- ++(.1,.25)
		coordinate (a) ++(-.1,-.25) |- +(.1,-.25) coordinate (b);
		\draw[-o] (b) to[resistor] ++(1,0) -- ++(.1,0) node[contact] {} --
		++(.5,0);
		\draw[-o] (a)
		to[resistor] ++(1,0) -- ++(.1,0) node[contact] {}
		to[resistor] ++(1,0) -- ++(.1,0) node[contact] {}
		-- ++(.5,0);
	\end{tikzpicture}
	\\
	&\pm
	\begin{tikzpicture}[semithick, baseline=-3pt, circuit ee IEC,
		every resistor/.style={circuit symbol size=width 2 height .75}]
		\draw (0,0) to[resistor] ++(1,0) |- ++(.1,.25)
		coordinate (a) ++(-.1,-.25) |- +(.1,-.25) coordinate (b);
		\draw[-o] (a) to[resistor] ++(1,0) -- ++(.1,0) node[contact] {} --
		++(.5,0);
		\draw (b) to[resistor] ++(1,0) |- ++(.1,.25)
		coordinate (b1) ++(-.1,-.25) |- +(.1,-.25) coordinate (b2);
		\draw[-o] (b1) to[resistor] ++(1,0) -- ++(.1,0) node[contact] {} --
		++(.5,0);
		\draw[-o] (b2) to[resistor] ++(1,0) -- ++(.1,0) node[contact] {} --
		++(.5,0);
	\end{tikzpicture}
	\pm
	\begin{tikzpicture}[semithick, baseline=-3pt, circuit ee IEC,
		every resistor/.style={circuit symbol size=width 2 height .75}]
		\draw (0,0) to[resistor] ++(1,0) |- ++(.1,.25)
		coordinate (b) ++(-.1,-.25) |- +(.1,-.25) coordinate (a);
		\draw[-o] (a) to[resistor] ++(1,0) -- ++(.1,0) node[contact] {} --
		++(.5,0);
		\draw (b) to[resistor] ++(1,0) |- ++(.1,.25)
		coordinate (b1) ++(-.1,-.25) |- +(.1,-.25) coordinate (b2);
		\draw[-o] (b1) to[resistor] ++(1,0) -- ++(.1,0) node[contact] {} --
		++(.5,0);
		\draw[-o] (b2) to[resistor] ++(1,0) -- ++(.1,0) node[contact] {} --
		++(.5,0);
	\end{tikzpicture}
	\\[3mm]
	&= 
	\begin{tikzpicture}[semithick, baseline=-3pt, circuit ee IEC,
		every resistor/.style={circuit symbol size=width 2 height .75}]
		\draw[-o] (0,0) to[resistor] ++(1,0) -- ++(.1,0) node[contact] {
		}to[resistor] ++(1,0) -- ++(.1,0) node[contact] {}
		to[resistor] ++(1,0) -- ++(.1,0) node[contact] {} -- ++(.5,0);
	\end{tikzpicture}
	\pm
	\begin{tikzpicture}[semithick, baseline=-3pt, circuit ee IEC,
		every resistor/.style={circuit symbol size=width 2 height .75}]
		\draw (0,0) to[resistor] ++(1,0) -- ++(.1,0) node[contact] {}
		to[resistor] ++(1,0) |- ++(.1,.25)
		coordinate (a) ++(-.1,-.25) |- +(.1,-.25) coordinate (b);
		\draw[-o] (a) to[resistor] ++(1,0) -- ++(.1,0) node[contact] {} --
		++(.5,0);
		\draw[-o] (b) to[resistor] ++(1,0) -- ++(.1,0) node[contact] {} --
		++(.5,0);
	\end{tikzpicture}
	\\
	&\pm 2 \times
	\begin{tikzpicture}[semithick, baseline=-3pt, circuit ee IEC,
		every resistor/.style={circuit symbol size=width 2 height .75}]
		\draw (0,0) to[resistor] ++(1,0) |- ++(.1,.25)
		coordinate (a) ++(-.1,-.25) |- +(.1,-.25) coordinate (b);
		\draw[-o] (a) to[resistor] ++(1,0) -- ++(.1,0) node[contact] {} --
		++(.5,0);
		\draw[-o] (b)
		to[resistor] ++(1,0) -- ++(.1,0) node[contact] {}
		to[resistor] ++(1,0) -- ++(.1,0) node[contact] {}
		-- ++(.5,0);
	\end{tikzpicture}
	\\
	&\pm 2 \times
	\begin{tikzpicture}[semithick, baseline=-3pt, circuit ee IEC,
		every resistor/.style={circuit symbol size=width 2 height .75}]
		\draw (0,0) to[resistor] ++(1,0) |- ++(.1,.25)
		coordinate (a) ++(-.1,-.25) |- +(.1,-.25) coordinate (b);
		\draw[-o] (a) to[resistor] ++(1,0) -- ++(.1,0) node[contact] {} --
		++(.5,0);
		\draw (b) to[resistor] ++(1,0) |- ++(.1,.25)
		coordinate (b1) ++(-.1,-.25) |- +(.1,-.25) coordinate (b2);
		\draw[-o] (b1) to[resistor] ++(1,0) -- ++(.1,0) node[contact] {} --
		++(.5,0);
		\draw[-o] (b2) to[resistor] ++(1,0) -- ++(.1,0) node[contact] {} --
		++(.5,0);
	\end{tikzpicture}
	~,
\end{align}
where the factors of $2$ appear due to the bracket being symmetric on odd
elements. Even though we do not aim to provide a rule for determining the
relative signs, we can see that every sum of order $k$ contains all the
connected tree diagrams with $k$ vertices
$\tikz[circuit ee IEC, semithick]
{\draw (0,0) -- ++(.3,0) node[contact] {} -- ++(.3,0);}$
that can be constructed according to the following rules:
\begin{enumerate}
	\item These combinations of elements are prohibited:
	\begin{equation}
		\begin{tikzpicture}[semithick, circuit ee IEC,
			every resistor/.style={circuit symbol size=width 2 height .75}]
			\draw (0,0) to[resistor] ++(.8,0) to[resistor] ++(.8,0);
		\end{tikzpicture}
		\quad
		\begin{tikzpicture}[semithick, circuit ee IEC,
			every resistor/.style={circuit symbol size=width 2 height .75}]
			\draw (0,0) to[resistor] ++(1,0);
			\draw[-o] (1,0) -- ++(.2,0);
		\end{tikzpicture}
		\quad
		\begin{tikzpicture}[semithick, circuit ee IEC, baseline=-3pt]
			\draw (0,0) -- ++(.25,0) node[contact] {} -- ++(.5,0) node[contact]
			{}-- ++(.25,0);
		\end{tikzpicture}
		\quad
		\begin{tikzpicture}[semithick, baseline=-3pt]
			\draw (0,0) -- ++(.25,0) |- ++(.25,.25) ++(-.25,-.25) |- +(.25,-.25)
			coordinate (a);
			\draw[-o] (a) -- ++(.2,0);
		\end{tikzpicture}
		\quad
		\begin{tikzpicture}[semithick, baseline=-3pt, circuit ee IEC]
			\draw (0,0) -- ++(.25,0) node[contact] {} -- ++(.4,0)
			|- ++(.25,.25) ++(-.25,-.25) |- +(.25,-.25);
		\end{tikzpicture}
		\quad
		\begin{tikzpicture}[semithick, baseline=-3pt, circuit ee IEC]
			\draw (0,0) -- ++(.25,0) |- ++(.25,.25) ++(-.25,-.25)
			|- ++(.25,-.25) node[contact] {} -- ++(.25,0);
		\end{tikzpicture}
		;
	\end{equation}
	\item The diagram has
		$\tikz[semithick, baseline=-3pt, circuit ee IEC, every
		resistor/.style={circuit symbol size=width 2 height .75}] {\draw (0,0)
		to[resistor] ++(1,0);}$
		as its root;
	\item The diagram has
		$\tikz[baseline=-3pt] {\draw[semithick,-o] (0,0) -- ++(.3,0);}$
		as every leaf.
\end{enumerate}

\begin{remark}
	As an immediate corollary of remark~\ref{rem:eig_cor} we see that diagrams
	beginning with
	\begin{equation}
		\begin{tikzpicture}[semithick, baseline=-3pt, circuit ee IEC,
			every resistor/.style={circuit symbol size=width 2 height .75}]
			\draw (0,0) to[resistor] ++(1,0) |- ++(.1,.25)
			coordinate (a) ++(-.1,-.25) |- +(.1,-.25) coordinate (b);
			\draw (a) to[resistor] ++(1,0);
			\draw (b) to[resistor] ++(1,0);
		\end{tikzpicture}
		~= \tfrac{1}{2} h_0 \br{h_0(-)}{h_0(-)}
	\end{equation}
	do not contribute to the eigenvalue corrections.
\end{remark}

\section{Conclusion}%
\label{sec:conclusion}

Quantum mechanics provides a playground for homological methods which can
further be applied to quantum field theories and string
theories~\cite{losev2019}. Thus, the utility of this work lies not in the
well-known formulas for perturbative corrections, but in paving the way for
further development of such methods~--- specifically in writing the
time-independent Schr\"odinger equation as a Maurer-Cartan equation and twisting
the differential. We see two directions inviting further investigation.

Firstly, energy level degeneracy makes the cohomology of $\Qt_0$
two-di\-men\-sional (four-dimensional, counting the ghost degrees of freedom),
and requires a different approach to calculate the corrections. We expect that
homotopy transfer and the perturbation lemma (see, for
instance,~\cite{crainic2004}) will play a crucial role there.

Secondly, in the spirit of our previous work~\cite{losev2023} we can introduce
symmetries of the Hamiltonian (conserved charges) into the differential. This
would most likely result in relations describing how the perturbation breaks
these symmetries.

\subsection*{Acknowledgements}%

This work was done under support of the grant №075-15-2022-289 for creation and
development of Euler International Mathematical Institute.

\end{document}